\documentclass[12pt]{article}
\usepackage[utf8]{inputenc}
\usepackage{amsmath}
\usepackage{amssymb}
\usepackage{amsthm}
\usepackage{xpatch}
\usepackage[width=16cm]{geometry}
\usepackage[dvipsnames]{xcolor}
\newcommand{\be}{\begin{equation}}
\newcommand{\bea}{\begin{eqnarray}}
\newcommand{\eea}{\end{eqnarray}}
\newcommand{\ee}{\end{equation}}
\newcommand{\qdet}{{\rm qdet}}
\newcommand{\qM}{\mathfrak{M}}
\newcommand{\clT}{T}
\newcommand{\opT}{\mathbf{T}}
\newcommand{\clM}{M}
\newcommand{\opM}{\mathbf{M}}
\newcommand{\clS}{S}
\newcommand{\opS}{\mathbf{S}}
\newcommand{\Sact}{A}
\newcommand{\Sg}{G}
\newcommand{\clU}{U}
\newcommand{\opU}{\mathbf{U}}

\newcommand{\ri}{{\mathsf i}}
\newcommand{\rd}{{\mathsf d}}

\newcommand{\Yang}{Y}

\newcommand{\ER}{Einstein-Rosen model}
\newcommand{\PPR}{Paternain, Perasa and Reisenberger}

\newtheorem{proposition}{Proposition}[section]

\newtheorem{definition}{Definition}[section]
\numberwithin{equation}{section}
\author{Boris A. Runov}
\title{On quantum determinants in integrable quantum gravity}
\begin{document}
\baselineskip 16pt

\medskip
\begin{center}
\begin{Large}\fontfamily{cmss}
\fontsize{17pt}{27pt}
\selectfont
	\textbf{On quantum determinants in integrable quantum gravity}
	\end{Large}

\bigskip \bigskip
\begin{large} 
B. Runov$^{a, b}$\footnote{e-mail:boris.runov@concordia.ca}
 \end{large}
 \\
\bigskip
\begin{small}
$^{a}${\em Centre de recherches math\'ematiques, Universit\'e de Montr\'eal, \\C.~P.~6128, succ. centre ville, Montr\'eal, QC H3C 3J7  Canada}\\
$^{b}${\em Department of Mathematics and Statistics, Concordia University\\ 1455 de Maisonneuve Blvd.~W.~Montreal, QC H3G 1M8  Canada}
\end{small}
 \end{center}
\medskip
	\begin{abstract}
		Einstein-Rosen waves with two polarizations are cylindrically symmetric solutions to vacuum Einstein equations. 
		Einstein equations in this case reduce to an integrable system. In 1971, Geroch has shown that this system 
		admits an infinite-dimensional group of symmetry transformations known as the Geroch group. 
		The phase space of this system can be  parametrized by a matrix-valued function of spectral parameter, called monodromy matrix.
		The latter admits the Riemann-Hilbert factorization into a pair of transition matrices, i.e. matrix-valued functions of spectral parameter such that 
		one of them is holomorphic in the upper half-plane, and the other is holomorphic in the lower half-plane.
		The classical Geroch group preserves the determinants of transition and monodromy matrices by construction.

		The algebraic quantization of the quadratic Poisson algebra generated by transition matrices of 
		Einstein-Rosen system  was proposed by Korotkin and Samtleben in 1997.
		Paternain, Perasa and Reisenberger have recently suggested a quantization of the Geroch group, 
		which can be considered as a symmetry of the quantum algebra of observables.
		They have shown that  commutation relations involving quantum monodromy matrices are preserved 
		by  the action of the  quantum Geroch group. 
		
		In present paper we introduce the notion of the determinant of the quantum monodromy matrix.
		We derive a factorization formula expressing the quantum determinant of the monodromy matrix as a 
		product of the quantum determinants of the transition matrices.
		The action of the quantum Geroch group is extended from the subalgebra generated by the monodromy matrix
		onto the full algebra of observables.
		This extension is used to prove that the quantum determinant of the quantum monodromy matrix is  invariant
		under the action of quantum Geroch group.
	\end{abstract}
\section{Introduction}
The Einstein-Rosen  model with two polarizations
describes cylindrically symmetric vacuum solutions in general relativity.
We shall refer to this model simply as \ER\  throughout this paper, using the term reduced Einstein-Rosen model for the system with one polarization originally discussed by Einstein and Rosen.
The reduced Einstein-Rosen model  can be considered \cite{AP1,Ku1} as an exactly solvable midi-superspace model of quantum gravity.
The full Einstein-Rosen model has  recently attracted renewed interest \cite{R3,R2,R1}  as it's Poisson structure was demonstrated to be closely related to Poisson brackets of null initial data for the full vacuum general relativity  in four dimensions. 
Therefore the quantization of the \ER\ might lead to  insights into more general quantum general relativity.

On classical level, integrability of the Einstein-Rosen model was  established by Belinsky and Zakharov in \cite{BZ},
and independently by  Maison \cite{Ma3}.
The group of non-commutative symmetries of the model was discovered by Geroch \cite{G1} and further explored Breitenlohner and Maison \cite{Ma2} and other authors.
The Poisson structure of \ER\ was explored in a series of works \cite{KS1,KS2,KS3,KS4,NKS1,NKS2,NK1}. 
The phase space of the system can be parametrized by so-called {\it transition matrices}, i.e. two one-parametric families of $SL(2,\mathbb{C})$-valued constants of motion $\clT_{\pm}(u)$, where complex parameter $u$ is called the {\it constant spectral parameter}. 
The physical content of the theory is encoded into the monodromy matrix originally introduced in \cite{Ma1}
( referred to as ``deformed metric'' in  \cite{R1}), which is the product of transition matrices computed at real values of spectral parameter.
It can be shown that the monodromy matrix can be identified with a submatrix of the metric of the four-dimensional spacetime at the symmetry axis\cite{Ma1}.

In \cite{KS1,KS2} the authors suggested a quantum deformation of the Poisson structure 
by proposing a set of exchange relations defining the quantum algebra of observables which quantize the classical Poisson algebra generated by transition matrices.

To complete the quantization program one has to specify a 
``physically acceptable'' (in the main body of the paper we discuss the meaning of the term)
representation of the algebra, which proved to be an elusive task. The representation  originally proposed by Korotkin and Samtleben \cite{KS2} doesn't possess a positive-definite scalar product.
On the other hand, in the alternative representation suggested in  the subsequent work by Niedermaier and Samtleben \cite{NS1} the quantum metric becomes non-symmetric, while the problem of positive-definitness remains unsolved.

Another approach to the problem of constructing the physically meaningful representation was suggested by \PPR\ \cite{R1}. 
Their approach is based on the fact that the Geroch group, which is a group of dressing transformations of the classical model, 
acts transitively on the space of solutions.
The transformations from the Geroch group are not canonical. Instead, the  Geroch group is a Lie-Poisson group.
Therefore, the elements of the Geroch group belong to an auxiliary phase space equipped with its own non-trivial Poisson bracket.
In accordance with the general theory of the quantum groups \cite{B1,BB1,BB2}, the quantization of the Geroch group should be a quasitriangular Hopf algebra.
In the approach proposed in \cite{R1} one defines the representation of the original algebra using the orbits of  quantum Geroch group.
The classical Geroch group is known to act transitively on the space of solutions with appropriate asymptotic behaviour \cite{Ma2}. 
Therefore it is natural to assume that the quantum analogue of the Geroch group would have the entire hypothetical Hilbert space as its orbit. 
The quantum analogue of the Lie-Poisson action constructed in \cite{R1} is an isomorphism of the algebra of observables. However, to be a genuine symmetry the quantum Geroch group must also preserve the essential properties of the representation which arise from the classical limit. In particular, on the real line the quantum monodromy matrix should be symmetric and positive-definite and its components must be Hermitian. Furthermore, the quantum determinants of the quantum transition matrices $\opT_{\pm}(u)$ must coincide with the identity operators for all $u$ in $\mathbb{C}$. Finally, the transition matrix $\opT_{+}(u)$ ($\opT_{-}(u)$) must me holomorphic in upper (resp. lower) half-plane.

The authors of \cite{R1} have proven that their candidate for the quantum Geroch transformation respects most of the key properties of the quantum monodromy matrix, such as its exchange relation with itself, positive-semidefinitness, 
invariance with respect to conjugation and the vanishing of antisymmetric part. 
However, the paper \cite{R1} does not contain a definition of the action of the quantum Geroch group on individual transition matrices $\opT_{\pm}$. Therefore, the compatibility of the quantum Geroch group symmetry with the requirement that the quantum determinants of $\opT_{\pm}$ coincide with the identity operator (which is an important property of the hypothetical physically meaningful representation) remains unchecked.
In this paper we address this issue. 

We introduce the quantum analogue of determinant of the monodromy matrix as follows.
\begin{definition}
The regularized quantum determinant of $\opM(u)$ is defined by the identity
\be
	\qdet \opM(u)\delta_{ij} =\frac{2}{3}{\rm Res}_{s=0}\epsilon_{jc}\epsilon_{ab} \opM^{ia}(u) \opM^{cb}(u+\ri\hbar+\ri s)\,,
\ee
where $\epsilon_{ab}$ is the completely antisymmetric tensor and the summation over repeated indices is assumed.
\end{definition}
The main result of the paper is  the following factorization formula for $\qdet M(u)$:
\be
\label{qdetMFact}
	\qdet \opM(u)= \qdet \opT_{+}(u)\, \qdet \opT_{-}(u)\,.
\ee
The quantum determinants of the transition matrices in the right hand side are central elements of the algebra of observables defined as
\be
	\qdet \opT_{\pm}(u)=\opT^{11}_{\pm}(u)\opT^{22}_{\pm}(u+\ri\hbar)-\opT^{12}_{\pm}(u)\opT^{21}_{\pm}(u+\ri\hbar)\,.
\ee
The identity (\ref{qdetMFact}) implies that the quantum determinant of the monodromy matrix belongs to the center of the  algebra of observables.
In physically relevant representation one must have
\be
	\qdet \opM(u)=1\,.
\ee
The relation (\ref{qdetMFact}) also means that the invariance of the regularized quantum determinant of the monodromy matrix under the action of the quantum Geroch group is a necessary condition of the invariance of the quantum determinants of the transition matrices. 
The latter statement partially answers the question posed in paper \cite{R1}.

We extend the action of the quantum Geroch group of \cite{R1} to the full algebra of observables as follows:
\bea
\label{GerochTpInt}
	\opT_{+}(u)\mapsto \tilde{\opT}_{+}(u)&:=&\opS(u\pm\textstyle{\frac{\ri}{2}}\hbar) \opT_{\pm}(u)\check{\opU}(-u+\textstyle{\frac{\ri}{2}}\hbar)\,,\\[2mm]
\label{GerochTmInt}
\opT_{-}(u)\mapsto \tilde{\opT}_{-}(u)&:=& \opS(u\pm\textstyle{\frac{\ri}{2}}\hbar) \opT_{\pm}(u)\opU^{T}(-u-\textstyle{\frac{\ri}{2}}\hbar)\,.
\eea
This transformation quantizes the action of the classical Geroch group, which is given by
\bea
\label{GerochTpClI}
	\clT_{+}(u) &\mapsto& \tilde{T}_{+}(u):=\clS(u)\clT_{+}(u)\clU^{-1}\left(\clS,\clT_{\pm}\mid u\right)\,,\\ 
[2mm] 
\label{GerochTmClI}
T_{-}(u) &\mapsto& \tilde{\clT}_{-}(u):=\clS(u)\clT_{-}(u)\clU^{T}\left(\clS,\clT_{\pm} \mid u\right)\,.
\eea
In the formulae (\ref{GerochTpClI}),(\ref{GerochTmClI}) symbol $\clS(u)$ denotes an element of the classical Geroch group parametrized by a $SL(2,\mathbb{R})$-valued function of the constant spectral parameter on the real line.
The symbol $\clU\left(\clS,\clT_{\pm} \mid u \right)$ denotes another $SL(2,\mathbb{R})$-valued function of the spectral parameter, which is also a functional of the transition matrices and the Geroch group element.
The superscript $T$ denotes the transposition.
The matrix-valued function $\clU(\clS,\clT_{\pm} \mid u)$ is fixed by the requirement that the transformed transition matrices $\tilde{\clT}_{\pm}(u)$ are analytical in the same half-planes as the untransformed ones and have the same leading term of the asymptotic expansion at large values of the spectral parameter.

The quantum matrices $\opS(u)$ and $\opU(u)$ in the formulae (\ref{GerochTpInt}),(\ref{GerochTmInt}) obey the exchange relations of $\mathfrak{sl}(2)$ Yangian:
\bea
	R(u-v)\overset{1}{\opS}(u)\overset{2}{\opS}(v)&=&\overset{2}{\opS}(v)\overset{1}{\opS}(u) R(u-v)\,,\\
	R(u-v)\overset{1}{\opU}(u)\overset{2}{\opU}(v)&=&\overset{2}{\opU}(v)\overset{1}{\opU}(u) R(u-v)\,.
\eea
The check denotes the ``quantum comatrix'', i.e.
\be
	\check{\opU}^{ij}(u)=\epsilon_{ia}\epsilon_{jb}\opU^{ji}(u)=\qdet \opU(u) \left(\opU^{-1}(u-\ri\hbar)\right)^{ij}\,.
\ee
The entries of $\opS(u)$ and $\opU(u)$ act not on the Hilbert space $\mathcal{H}$ but on their own respective phase spaces $\mathcal{G}$ and $\mathcal{K}$ so that the following commutation relations are respected
\bea
\label{SUcommInt}
	[\opS^{ij}(u),\opT^{kl}_{\pm}(v)] &=& 0 \,, \\[1mm]
	[\opS^{ij}(u),\opU^{kl}_{\pm}(v)] &=& 0 \,, \\[1mm]
	[\opU^{ij}(u),\opT^{kl}_{\pm}(v)] &=& 0 \,.
\eea
The operator $\opU(u)$ is determined by the operators $\opS(u)$ and $\opT(u)$. However, the explicit relation between them is unknown to the author.
Unlike the classical  matrices $\tilde{\clT}_{\pm}(u)$, the quantum operators $\tilde{\opT}_{\pm}(u)$ can't be
made analytical in their respective half-planes (see Proposition \ref{prop_nogo_hol}).
The best we can hope for is  that the representation of the algebra generated by transformed transition matrices on the space $\,\mathcal{G}\otimes\mathcal{H}\otimes\mathcal{K}\,$ is decomposable into a direct sum of several other representations, of which at least one is ``physically acceptable''. In other words, in some appropriate basis the transformed transition matrices are block-diagonal, and some blocks possess the required analytical properties. We hypothesize that the operator $\opU(u)$ can be fixed completely by requiring that such decomposition is possible.
%

We show that the transformation (\ref{GerochTpInt}),(\ref{GerochTmInt}) preserves the quantum determinants of the transition matrices. Then the factorization formula (\ref{qdetMFact}) implies that the quantum Geroch group also preserves the  quantum determinant of the monodromy matrix.
However,  the derivation of the identity (\ref{qdetMFact}) relies upon the analytical properties of the quantum transition matrices and their products in ``physically acceptable'' representations.
To prove that there exists a representation of the Yangian such that with the operator $\opU(u)$ in this representation the transformed transition matrices $\opT_{\pm}(u)$ have regular blocks described above one needs to know more about representations of the algebra of observables $\mathcal{D}$ and the Yangian generated by the operator $\opS(u)$. We hope to address this problem in future works.


The paper is organized as follows.
In section \ref{section_class} we review the integrable structure of the classical \ER\ following \cite{KS1} and \cite{R1}. In section \ref{section_rep} we describe the quantum algebra of observables and formulate restrictions on its physically relevant representation.
In section \ref{section_qdet} we introduce the quantum determinant of the monodromy matrix and derive the factorization formula (\ref{qdetMFact}). In section \ref{section_geroch} a generalization of the action of the quantum Geroch group of \PPR\ \cite{R1} to the full algebra of observables is proposed.

\section{Integrable structure of the classical \ER}
\label{section_class}
\subsection{Equations of motion and Poisson structure}
Let us consider the metric of the form
\be
	\rd s^2=e^{\Gamma(\rho,\tau)}\left(-d\tau^2+d\rho^2\right)+g_{ab}(\rho,\tau)\rd \bar{x}^a \rd \bar{x}^b\,,\quad a,b=1,2
\ee
where the function $\Gamma$ and the $SL(2,\mathbb{R})$-valued symmetric matrix $g$ are independent of $\bar{x}_1,\bar{x}_2$. Then the vacuum Einstein equations  imply
that the matrix $g$ satisfies the following equation
\be
\label{ErnstEq}
	\partial_{\rho}\left(\rho g^{-1}\partial_{\rho}g\right) -\partial_{\tau}\left(\rho g^{-1}\partial_{\tau}g\right)=0\,,
\ee
and that the conformal factor $\Gamma(\rho,\tau)$ is given by
\be
\label{ErnstH}
	\Gamma(\rho,\tau)=\frac{1}{2}\int\limits_{0}^{\rho}\rho^{\prime}d\rho^{\prime} 
	\text{Tr}\left(\left(g^{-1}\partial_{\rho^{\prime}}g\right)^2+\left(g^{-1}\partial_{\tau}g\right)^2\right)\,.
\ee
The solutions of the equations (\ref{ErnstEq}),(\ref{ErnstH}) are known as the Einstein-Rosen waves with two polarizations.
The equation (\ref{ErnstEq}) can be derived from the following action:
\be
\label{ActG}
	\Sact=\frac{1}{2}\int\limits_{0}^{\infty}\rho d\rho 
	\text{Tr}\left(\left(g^{-1}\partial_{\rho^{\prime}}g\right)^2-\left(g^{-1}\partial_{\tau}g\right)^2\right)\,.
\ee
The action (\ref{ActG}) differs  from the action of the non-linear sigma model with $PSL(2,\mathbb{R})$ symmetry only by the coordinate-dependent factor $\rho$. Therefore one can think of this model as of a non-autonomous generalization of the sigma model. 
We will refer to this theory as the \ER.
A similar metric with a different signature  describes the cosmological Gowdy model. Integrable structure outlined below is also present there with minor modifications \cite{KS1}.

While the equation (\ref{ErnstEq}) is written in terms of the matrix $g$, which is a submatrix of the metric,
it is convenient to rewrite this equation in terms of the zweibein to stress the similarity to coset sigma models.
 To do this, we introduce the flat connection $J$ as follows:
\be
	J_{\mu}=\mathcal{V}^{-1}\partial_{\mu} \mathcal{V}\,.
\ee
The components of the connection $J$ take values in the algebra $\mathfrak{g}=\mathfrak{sl}_2(\mathbb{R})$, which can be decomposed into a direct sum of the algebra $\mathfrak{h}=\mathfrak{so}_2(\mathbb{R})$ and the orthogonal (i.e. symmetrized) component $\mathfrak{l}$
\be
	\mathfrak{g}=\mathfrak{h}\oplus\mathfrak{l}\,.
\ee
Then the components of the current $J$ can be decomposed as
\be
	J_{\mu}=P_{\mu}+Q_{\mu}\,,\quad P_{\mu}=\frac{1}{2}\left(J_{\mu}+J_{\mu}^{T}\right)\in\mathfrak{l}\,,\;Q_{\mu}=\frac{1}{2}\left(J_{\mu}-J_{\mu}^{T}\right)\in\mathfrak{h}\,.
\ee
The equation (\ref{ErnstEq}) then takes the form
 \be
	D_{\mu}(\rho P_{\mu})=0\,,\quad D_{\mu}P_{\nu}=\partial_{\mu}P_{\nu}-[P_{\mu},Q_{\nu}]\,.
 \ee
The action (\ref{ActG}) in this parametrization looks as follows:
 \be
 \label{actCur}
	\Sact=\int d\tau \int_{0}^{\infty}\rho d\rho \left(P_0^2-P_1^2\right)\,.
\ee
Expressing the components of the current $J$ in terms of the zweibein one finds
\be
\label{ActV}
	\Sact=\frac{1}{2}\int\limits_{0}^{\infty}\rho d\rho {\rm Tr} \left[D_{\mu}\mathcal{V}\mathcal{V}^{-1}D^{\mu}\mathcal{V} \mathcal{V}^{-1}\right]\,.
\ee
The action written in terms of the zweibein is invariant under local $SO(2)$ transformations. The physical fields therefore take values in the coset space $SL(2,R)/SO(2)$.

The components of the metric $g_{ab}$ are given by
\be
\label{met_zweibein}
	g_{ab}=\mathcal{V}_{a}^{i}\mathcal{V}_{b}^{j}\delta_{ij}.
\ee

As it was discovered in \cite{BZ,Ma1} the equation (\ref{ErnstEq}) is integrable in the sense of theory of solitons (i.e. it possesses a zero curvature representation). Namely, consider the following deformation of the connection $J$ which  also turns out to be flat on solutions of equations of motion
\be
\label{defcurr}
	\hat{J}_{\mu}=Q_{\mu}+\frac{1+\gamma^2}{1-\gamma^2}P_{\mu}+\frac{2\gamma}{1-\gamma^2}\epsilon_{\mu\nu}P^{\nu}
\ee
provided the parameter $\gamma$, which is called the {\it dynamical spectral parameter}, depends on spacetime coordinates and the {\it constant spectral parameter} $u$ as follows
\be
\label{gammaDef}
	\gamma(\rho,\tau,u)=\frac{1}{\rho}\left(u+\tau+\sqrt{(u+\tau)^2-\rho^2}\right)\,.
\ee
For a fixed point in the spacetime, the dynamical spectral parameter is in fact a function on a two-sheeted covering of the Riemann sphere.
Either choice of the branch of the square root in the definition (\ref{gammaDef}) results in a flat connection $\hat{J}$ as long as the choice is consistent, i.e. the dynamical spectral parameter is a smooth function of the spacetime coordinates for $\rho>0$.

The solution $\hat{\mathcal{V}}(\rho,\tau,u)$ of the linear problem
\be
	\partial_{\mu}\hat{\mathcal{V}}=\hat{\mathcal{V}}\hat{J}
\ee
can be considered the ``deformed zweibein''. We normalize the matrix $\hat{\mathcal{V}}$ as follows:
\be
	\hat{\mathcal{V}}({\bf x},u)=\mathcal{V}({\bf 0})
	\mathcal{P} \exp\left[\int\limits_{{\bf 0}}^{{\bf x}}\hat{J}_{\mu}({\bf y},u) \rd y^{\mu}\right]\,,
\ee
where
\be
	{\bf x}=(\rho,\tau)\,,\quad {\bf 0}=(0,0)\,.
\ee
Similarly to (\ref{met_zweibein}) one can construct the ``deformed metric'' (borrowing the terminology of \cite{R1}) from the ``deformed zweibein'':
\be
	\clM({\bf x},u)=\hat{\mathcal{V}}({\bf x},u)\kappa\left(\hat{\mathcal{V}}({\bf x},u)\right)
\ee
where the involution $\kappa$ interchanges the sheets of the Riemann surface on which the dynamical spectral parameter $\gamma$ is defined,
i.e.
\be
	\kappa(\gamma({\bf x},u))=
	\frac{1}{\rho}\left(u+\tau-\sqrt{(u+\tau)^2-\rho^2}\right)= \gamma^{-1}({\bf x},u)\,.
\ee
As it is explained in \cite{R1}, one can prove that $\clM({\bf x},u)$ does not in fact depend on ${\bf x}$.
So we will call it the {\it monodromy matrix} and denote it by $\clM(u)$.

Korotkin and Samtleben \cite{KS1,KS2} found two families of conserved quantities $T_{\pm}(u,\tau)$
\be
\label{Tdefcl}
	\clT_{\pm}(u,\tau)=\lim\limits_{\epsilon\rightarrow 0}
	\mathcal{P}
	\exp \left\lbrace 2 \int\limits \rd\rho\,
	\left(\frac{\gamma^2_{\pm}g^{-1}\partial_{\rho}g}{1-\gamma^{2}_{\pm}}-\frac{\gamma^{2}_{\pm}g^{-1}\partial_{\tau}g}{1-\gamma^{2}_{\pm}}\right)
	\right\rbrace\,,
\ee
with 
\be
	\gamma_{\pm}=\gamma(\rho,\tau,u\pm i\epsilon)\,.
\ee
Since the total  derivative of $\clT_{\pm}(u,\tau)$ with respect to the time $\tau$ vanishes  due to the flatness of $\hat{J}$ \cite{KS1},
we shall denote these matrices by $\clT_{\pm}(u)$.
The matrices $\clT_{\pm}(u)$ solve the Riemann-Hilbert problem for the jump matrix $\clM(u)$ on the real line \cite{KS1}, 
i.e. they provide a factorization of the matrix $\clM(u)$
\be
	\clM(u)=\clT_{+}(u)\clT_{-}^{T}(u)
\ee
such that $\clT_{+}$ is holomorphic in the upper half-plane, and $\clT_{-}$ is holomorphic in the lower half-plane.
The set of matrices $\clT_{\pm}(u)$ completely determines the metric. 
The equations of motion (\ref{ErnstEq})  imply that
\be
	\frac{d}{d\tau} \clT_{\pm}(u)=0\,.
\ee

The Poisson brackets derived from the action (\ref{ActG}) ignoring the constraint $\text{det} g=1$ look like
\be
	\{ g_{ab}(\rho), \pi_{cd}(\rho')\}=\delta_{ac}\delta_{bd}\delta(\rho-\rho')
\ee
with canonical momenta given by
\be
	\pi_{ab}(\rho)=\rho \left(g^{-1}\partial_{\tau}g g^{-1}\right)_{ab}\,.
\ee
Taking the constraints into account one can compute the Poisson brackets between the components of the current. 
The non-vanishing brackets  are given by
\bea
\label{PbCurr}
	\{\overset{1}{P}_0(\rho),\overset{2}{P}_1(\rho')\} &=& \rho^{-1}\,[\Omega_{\mathfrak{l}}, \overset{2}{Q}_1(\rho)]\,\delta(\rho-\rho')+
	\rho^{-1}\,\partial_{\rho}\delta(\rho-\rho')\,\Omega_{\mathfrak{l}}\, \\[1mm]
	\{\overset{1}{P}_0(\rho),\overset{2}{Q}_1(\rho')\} &=& \rho^{-1}\,[\Omega_{\mathfrak{l}}, \overset{2}{P}_1(\rho)]\,\delta(\rho-\rho')\,,
\eea
where the following notation was adopted for the Poisson brackets between matrices
\be
	\{ \overset{1}{A},\overset{2}{B}\}^{ab,cd}=\{A^{ab},B^{cd}\}\,,
\ee
and the symbol $\Omega_{\mathfrak{l}}$ is given by:
\be
	\Omega_{\mathfrak{l}}=\mathbf{1}\otimes\mathbf{1}+\sigma_{x}\otimes\sigma_{x}+\sigma_{z}\otimes\sigma_{z}\,.
\ee
From the Poisson brackets  (\ref{PbCurr}) we can derive the Poisson brackets between the transition matrices:
\be
\label{PBsame}
	\{\overset{1}{\clT_{\pm}}(u),\overset{2}{\clT}_{\pm}(v)\}=[r(u-v),\overset{1}{\clT}_{\pm}(u)\overset{2}{\clT}_{\pm}(v)]\,,
\ee
\be
\label{PBtwist}
	\{\overset{1}{\clT}_{\pm}(u),\overset{2}{\clT}_{\mp}(v)\}= \left(r(u-v)\overset{1}{\clT}_{\pm}(u)\overset{2}{ \clT}_{\mp}(v)
	-\overset{1}{\clT}_{\pm}(u)\overset{2}{\clT}_{\mp}(v)r^{\eta}(u-v)\right)\,.
\ee

The brackets (\ref{PBsame}),(\ref{PBtwist})  have $r$-matrix structure with rational classical $r$-matrix
\be
\label{cRmat}
	r(u)=\frac{1}{u}E_{ab}\otimes E_{ba}\,,\quad r_{12}^{\eta}=\frac{1}{u}I- r_{12}^{T}(u)
\ee
satisfying the classical Yang-Baxter equation
\be
\label{cYBE}
	[\overset{12}{r}(u_1-u_2),\overset{13}{r}(u_1-u_3)]+[\overset{23}{r}(u_2-u_3),\overset{13}{r}(u_1-u_3)]+[\overset{12}{r}(u_1-u_2),\overset{23}{r}(u_2-u_3)]=0\,.
\ee
The equation (\ref{cYBE}) is defined on the tensor product of three identical spaces $V\otimes V \otimes V$ with $V=\mathbb{C}^2$, and the indices on top of the symbols $r$, $r^{\eta}$ denote the pairs of  spaces on which the respective matrix acts non-trivially. The symbol $T$ in (\ref{cRmat}) denotes the transposition with respect to  one of the spaces.
\subsection{Geroch group}
The space of solutions of (\ref{ErnstEq}) is isomorphic to the space of admissible boundary values of metric on the symmetry axis, 
which within the construction of \cite{KS1} are given by
\be
	g(0,u)=\clM(u)\,.
\ee
The admissible boundary values are regular functions on the real axis taking values in $2\times2$ symmetric real matrices with unit determinant.
Any admissible boundary value can be obtained from the trivial with help of the following transformation 
\be
\label{GerochMclI}
	g(0,u)\mapsto \tilde{g}(0,u):=\clS(u)g(0,u)\clS^{T}(u)\,,
\ee
where the parameter of the transformation $\clS(u)$ is  matrix-valued function of a real argument. 
The function $\clS(u)$ takes values in $SL(2,\mathbf{R})$ and tends to the identity matrix for large absolute values of $u$.
The transformations (\ref{GerochMclI}) with different choices of the function $\clS(u)$ obviously form a group.
The action of the transformation (\ref{GerochMclI}) on boundary values can be obviously extended to the action on fields such that any solution maps to a solution. 
This group of transformations is known as the Geroch group, and it is a particular example of a dressing transformation commonly found in infinite-dimensional classical integrable models.
It is not a symmetry in the usual sense of the Hamiltonian mechanics: it does not preserve the Hamiltonian. Contrarily, it is a Lie-Poisson group.

The infinitesimal action of the Geroch group on an  arbitrary observable  can be represented as follows.
Let the contour $l_{+}$ be the line $\Im(u)=\varepsilon$, and the contour $l_{-}$ be the line $\Im(u)=-\varepsilon$ for some small positive value of $\varepsilon$.
Let us define the adjoint action of the transition matrices by
\be
	\text{ad}_{\clT_{\pm}(w)}X=\{\clT_{\pm}(w),X\}\,.
\ee
Let $\Lambda(w)$ be a regular $\mathfrak{sl}(2,\mathbb{R})$-valued function of a real argument representing an element of the Lie algebra of the Geroch group.
Then the action of the algebra element $\Lambda(w)$ on an arbitrary observable is given by
\be
\label{GerochActArb}
	\Sg[\Lambda]=\lim\limits_{\varepsilon\rightarrow+0}\text{Tr}\left(
	\int\limits_{l_{+}}\clT_{+}^{-1}(w)\Lambda(w)\text{ad}_{\clT_{+}(w)}dw
	+\int\limits_{l_{-}}\clT_{-}^{-1}(w)\Lambda(w)\text{ad}_{\clT_{-}(w)}dw
	\right)\,.
\ee

The formula (\ref{GerochActArb})  implies \cite{KS1} that classically the Geroch group acts on the monodromy matrix $\clM(u)$ as 
\be
\label{GerochActMinf}
	\Sg[\Lambda]\clM(u)=\Lambda(u)\clM(u)+\clM(u)\Lambda^{T}(u)\,.
\ee
The action (\ref{GerochActMinf})  can obviously be exponentiated to
\be
\label{GerochMCl}
	\clM(u)\mapsto \tilde{\clM}(u):=\clS(u)\clM(u)\clS(u)^{T}\,,
\ee
where the matrix $\clS(u)$ belongs to $SL(2,\mathbb{R})$.
The action of the Geroch group on  the transition matrices can be written as
\bea
	\clT_{+}(u) &\mapsto& \tilde{T}_{+}(u):=\clS(u)\clT_{+}(u)\clU^{-1}\left(\clS,\clT_{\pm}\mid u\right)\,,\\[1mm]
	\clT_{-}(u) &\mapsto& \tilde{T}_{-}(u):=\clS(u)\clT_{-}(u)\clU^{T}\left(\clS,\clT_{\pm}\mid u\right)\,,
\eea
where the ``gauge'' factor $\clU\left(\clS,\clT_{\pm}\mid u\right)$ is a functional of the transition matrices and the element of the Geroch group. The matrix $\clU\left(\clS,\clT_{\pm} \mid u \right)$ is implicitly fixed by the following requirements.
Firstly, it is a function of a real argument which takes values in the group $SL(2,\mathbb{R})$. Secondly, the transformed transition matrices $\tilde{\clT}_{\pm}(u)$ admit analytical continuation into the same respective half-planes as the original ones.
Finally, the matrix-valued function $\clU\left(\clS, \clT_{\pm} \mid u\right)$ tends to the identity matrix at large values of the spectral parameter $u$.

The action of the Geroch group preserves the Poisson structure in the following sense: if we consider  independent elements of the  matrices $\clS(u)$
to be dynamical variables with a nontrivial Poisson structure
\be
	\{\overset{1}{\clS}(u),\overset{2}{\clS}(v)\}=[r(u-v),\overset{1}{\clS}(u)\overset{2}{\clS}(v)]\,,
	\quad \{\overset{1}{S}(u),\overset{2}{\clT}_{\pm}(v)\}=0
\ee
then the Poisson bracket between the transformed matrices $\tilde{\clT}_{\pm}$ would be given by the same formulae (\ref{PBsame}),(\ref{PBtwist}). 
\section{Operator products and properties of representation}
\label{section_rep}
In \cite{KS1} the authors suggested a consistent way to quantize the Einstein-Rosen model. 
The Poisson brackets (\ref{PBsame}),(\ref{PBtwist}) should be replaced by the following exchange relations:
\be
\label{RTTint}
	R(u-v) \overset{1}{\opT}_{\pm}(u)\overset{2}{\opT}_{\pm}(v)=\overset{2}{\opT}_{\pm}(v)\overset{1}{\opT}_{\pm}(u) R(u-v)\,,
\ee
\be
\label{Rpmint}
	R(u-v-\ri\hbar) \overset{1}{\opT}_{-}(u) \overset{2}{\opT}_{+}(v) =\overset{2}{\opT}_{+}(v) \overset{1}{\opT}_{-}(u) R^{\eta}(u-v+\ri\hbar) \chi(u-v)\,, 
\ee
where the scalar factor $\chi$ reads
\be
	\chi(u)=\frac{u(u-2\ri\hbar)}{u^2+\hbar^2}\,,
\ee
and the arguments in (\ref{Rpmint}) satisfy
\be
	u\notin \{v-\ri\hbar,v,v+\ri\hbar,v+2\ri\hbar\}\,.
\ee
The indices $1,2$ on top of the symbols $\opT_{\pm}$ above denote the space on which the corresponding matrix acts non-trivially:
\be
	\overset{1}{\opT}_{\pm}(u)=\opT_{\pm}(u)\otimes \mathbf{1}\,,\quad \overset{2}{\opT}_{\pm}(u)=\mathbf{1}\otimes \opT_{\pm}(u)\,.
\ee
The quantum $R$-matrix and its twisted version $R^{\eta}$ are given by
\be
	R(u,v)=(u-v)\left(\mathbf{1}-\ri\hbar r(u-v)\right)\,, \quad R^{\eta}(u,v)=-R^{T}(\ri\hbar-v,u)\,.
\ee
where the superscript $T$ denotes the transposition with respect to one of the spaces in the tensor product, i.e.
\be
	R^{T}_{ijkl}=R_{jikl}=R_{ijlk}\,.
\ee
The operators $\opT_{+}^{ij}(u)$ and $\opT_{-}^{ij}(u)$ are related via the Hermitian conjugation
\be
	\left(\opT_{+}^{ij}(u)\right)^{\dag}=\opT_{-}^{ij}(\bar{u})\,.
\ee
The algebra of observables is generated by the evaluations of the operator-valued functions $\opT_{\pm}^{ij}(u)$ on the real axis subject to the relations (\ref{RTTint},\ref{Rpmint}). 
Any polynomial in $\opT_{\pm}^{ij}(u)$, in particular the bilinear 
expressions of the form $\opT^{ij}_{+}(u)\opT^{kl}_{-}(v)$, is well-defined for any real value of $u$ and $v$, which  leads to the following identities
\be
\label{idealX}
	\opT^{ai}_{-}(u)\opT_{+}^{aj}(u)=0\,,
\ee
\be
\label{idealY}
	\opT^{1i}_{-}(u)\opT^{2j}_{+}(u)+\opT^{2i}_{-}(u)\opT^{2j}_{+}(u)=0\,.
\ee
The expressions in the left hand sides of the identities (\ref{idealX}), (\ref{idealY})
form an ideal and can be factored out.
Within the quotient algebra the relation (\ref{Rpmint}) can be extended to the point $u=v$.

We shall call this quotient algebra, following the paper \cite{KS1}, a ``twisted $\mathfrak{sl}(2)$ Yangian double'' and denote it by $\mathcal{D}$. The parameter $u$ will be referred to as the {\it spectral parameter}.

Similarly to the classical case one defines the {\it quantum monodromy matrix} as follows:
\be
	\opM(u)=\opT_{+}(u)\opT_{-}^{T}(u)\,.
\ee
The operators $\opM(u)$ form a subalgebra of the algebra $\mathcal{D}$, with the following exchange relations:
\be
	R(u-v)\overset{1}{\opM}(u)R^{\eta}(v-u+2\ri\hbar)\overset{2}{\opM}(v)=
	\overset{2}{\opM}(v)R^{\eta}(u-v+2\ri\hbar)\overset{1}{\opM}(u)R^{\eta}(v-u)\frac{\chi(u-v)}{\chi(v-u)}\,.
\ee

The following property of the algebra $\mathcal{D}$ was established in the work \cite{KS2}.
\begin{proposition}
	The {\it quantum determinants} of the transition matrices
\be
	\label{qdetDefB}
	\qdet \opT_{\pm}(u)=\opT^{11}_{\pm}(u)\opT^{22}_{\pm}(u+\ri\hbar)-\opT^{12}_{\pm}(u)\opT^{21}_{\pm}(u+\ri\hbar)
\ee
and the antisymmetric part of the quantum monodromy matrix
\be
	\opM^{12}(u)-\opM^{21}(u)
\ee
belong to the center of the algebra $\mathcal{D}$.
\end{proposition}

The RTT relation (\ref{RTTint}) defines two subalgebras: $\Yang_{+}$ generated  by the operators $\opT_{+}^{ij}(u)$, and $\Yang_{-}$ generated by the operators $\opT_{-}^{ij}(u)$.
The equation (\ref{RTTint}) coincides with the defining relation of the Yangian $Y(\mathfrak{sl}(2))$ up to a rescaling of the spectral parameter: 
\be
\label{RTTtextbook}
	R^{\mathfrak{gl}_2}(u-v)\overset{1}{\opT}(u)\overset{2}{\opT}(v)=\overset{2}{\opT}(v)\overset{1}{\opT}(u)R^{\mathfrak{gl}_2}(u-v)\,,
\ee
\be
	R(u)= u  R^{\mathfrak{gl}_2}(-u\ri \hbar^{-1})\,.
\ee
According to the textbook definition \cite{M1}, the Yangian relation (\ref{RTTtextbook}) should be interpreted as an equality of two bivariate formal power series
in inverse powers of $u$ and $v$. The symbol $\opT(u)$ should be understood as a formal generating function of the Yangian generators, which are given by the coefficients of the expansion of $\opT(u)$ in $u^{-1}$:
\be
\label{asympT}
	\opT^{ij}(u)=\delta_{ij}+\sum_{k=1}^{\infty} t^{(k)}_{ij} u^{-k}\,.
\ee
In particular representations of Yangian the evaluations  $\opT^{ij}(u)$ are allowed to be ill-defined (i.e. not be a part of the algebra)  for  some (or even all) values of $u$, but the generators $t_{ij}^{(k)}$ are always assumed to be well-defined (i.e. must have finite matrix elements).

The representation of the algebra $\mathcal{D}$ on the Hilbert space of the quantum \ER\ is unknown.
However, one can expect that in a hypothetical representation analytical and other properties of the quantum operators $\opT_{\pm}(u)$
reproduce respective properties of their classical analogues (discussed in  the previous section) in the limit $\hbar \rightarrow 0$.
It is therefore natural to look for representations satisfying the following conditions.

\begin{enumerate}
	\item \label{condEval}
		All matrix elements of the operators $\opT_{\pm}^{ij}(u)$ in this representation are smooth complex-valued functions on the real line.
	\item The Fourier transforms of the operators $\opT_{\pm}^{ij}(u)$
\be
	\hat{\opT}_{\pm}(k)=\int\limits_{-\infty}^{\infty} e^{\ri k u} \opT_{\pm}^{ij}(u) \rd u
\ee
are convergent for real $k\neq 0$. 
	\item The operators $\opT_{+}^{ij}(u)$ is analytic in the lower half-plane, and the operators $\opT_{-}^{ij}(u)$ must be analytic in the upper half-plane.
	\item On the imaginary axis the following asymptotic formulae hold true:
\bea
	\opT_{+}^{ij}(-\ri s)&=&\delta_{ij}+\sum_{n=1}^{\infty} (-1)^n t^{(n)}_{ij} \hbar^{n} s^{-n}\,\quad s\in \mathbb{R}\,, s\rightarrow +\infty\\
	\opT_{-}^{ij}(\ri s)&=& \delta_{ij}+\sum_{n=1}^{\infty}  t^{(-n)}_{ij} \hbar^{n} s^{-n}\,\quad s\in \mathbb{R}\,, s\rightarrow +\infty
\eea
Note that this requirement does not imply the existence of an open neighbourhood where both expansions are valid simultaneously. 
\item 
	The Fourier modes $\hat{\opT}_{\pm}^{ij}(k)$ have the following expansion
\bea
	\hat{\opT}_{+}^{ij}(k)&=&\delta_{ij}\delta(k)+\theta(k)\sum_{n=0}^{\infty}t^{(n+1)}_{ij}\hbar^{n+1}\frac{k^n}{n!}\\
	\hat{\opT}_{-}^{ij}(k)&=&\delta_{ij}\delta(k)+\theta(-k)\sum_{n=0}^{\infty}(-1)^{n} t^{(-n-1)}_{ij}\hbar^{n+1}\frac{k^n}{n!}
\eea
\item The operators $\opT_{+}^{ij}(u)$ and $\opT_{-}^{ij}(u)$ are analytic in the strip $-\frac{\hbar}{2}<\Im(u)<\frac{\hbar}{2}$, and  all normal-ordered products (i.e. the products such that all the instances of $\opT_{-}$ are to the right of all $\opT_{+}$) are regular in all arguments as long as all the arguments lie within this strip. However, the products with a different ordering, p.e. $\opT_{-}^{ij}(u)\opT_{+}^{kl}(u\pm \ri\hbar)$, might be singular.
\item
	\label{cent_cond}
	In order to achieve the consistency with the classical
expressions  the central elements are fixed as follows:
\be
	\qdet \opT_{+}(u)=\qdet \opT_{-}(u)=1\,,
\ee
\be
		\opM^{12}(u)-\opM^{21}(u)=0\,.
\ee
\end{enumerate}
We shall refer to  representations satisfying the conditions \ref{condEval} - \ref{cent_cond} as ``{\it physically acceptable}'' representations.
Let us summarize the desired properties of the hypothetical Hilbert space representation:
\begin{itemize}
	\item The Fourier transforms $\hat{\opT}_{\pm}(k)$ are well-defined for $k\neq 0$
	\item The operator $\opT_{+}(u)$ is analytic for $\Im(u)>-\hbar$ (negative frequencies vanish)
	\item The operator $\opT_{-}(u)$ is analytic for $\Im(u)<\hbar$ (positive frequencies vanish)
	\item $\text{qdet}\opT_{\pm}(u)=1$
	\item $\opM^{12}(u)=\opM^{21}(u)$
	\item All normal-ordered products of the operators $\opT_{\pm}(u)$ (i.e. such that all instances of $\opT_{-}(u)$ are to the right of all instances of $\opT_{+}(u)$) are regular functions in all arguments in the strip.
\end{itemize}
At this point we can't prove the existence of such a representation.
We proceed on the assumption that such a representation exists, leaving the construction for the future.

The following proposition shows that a ``physically acceptable'' representation can't be finite-dimensional.
\begin{proposition}
\label{FinRepProp}
There is no finite-dimensional representation satisfying all the constraints above.
\end{proposition}
\begin{proof}
For a finite-dimensional representation all products of regular operator-valued functions must be regular.
	Multiplying both sides of the identity (\ref{Rpmint}) by a factor $(u-v)^2+\hbar^2$ and evaluating both sides of the resulting equation at the points $u-v=\pm\ri\hbar$ one gets the following identity:
\be
\label{ihconstr}
	\opT_{+}^{ij}(u)\opT_{-}^{kl}(u\pm \ri \hbar)=0\,.
\ee
However, the product
\be
	\label{qdetProdNogo}
	\qdet \opT_{+}(u) \qdet \opT_{-}(u) =1
\ee
	is, per condition \ref{cent_cond},  non-zero.
	Let us  rewrite (\ref{qdetProdNogo}) in terms of the matrix entries of the matrices $\opT_{\pm}(u)$ using the definition (\ref{qdetDefB}). Then it is clear that this product is a linear combination of expressions proportional to quadratic expressions of the form (\ref{ihconstr}), and therefore must vanish in any finite-dimensional representation. Thus we have demonstrated that a finite-dimensional representation can't satisfy all the requirements simultaneously.
\end{proof}
\section{Quantum determinant of the deformed metric}
\label{section_qdet}
The quantum comatrix $\check{\mathbf{K}}(u)$ of a $2\times 2$ matrix $\mathbf{K}(u)$ with entries $\mathbf{K}^{ij}(u)$ being formal series with coefficients in Yangian is defined as
\be
	\check{\mathbf{K}}^{ij}(u)=\epsilon_{ia}\epsilon_{jb}\mathbf{K}^{ba}(u)\,.
\ee
Due to the RTT relation (\ref{RTTint})  the quantum transition matrices matrices $\opT_{\pm}(u)$ have the following property:
\be
\label{Tcomat}
	\opT_{\pm}(u) \check{\opT}_{\pm}(u+\ri\hbar)=\qdet \opT_{\pm}(u) \mathbf{1}\,,
\ee
where $\text{qdet}\opT_{\pm}(u)$ lies in the center of algebra due to the double being at the critical level and $\mathbf{1}$ is a $2\times 2$ identity matrix.
Since for the classical \ER\ we know that $\text{det}\clT_{\pm}(u)=1$, we are primarily interested in the representations such that
\be
\label{qdetCond}
	\qdet \opT_{\pm}(u)=1\,.
\ee

In order to verify that the transformation suggested in the work \cite{R1} as a quantum analogue of the Geroch group (see section \ref{section_geroch} for details) 
preserves the vacuum expectation values of the elements of the algebra $\mathcal{D}$ (and therefore preserves the physical content of the model)
one has to ensure that the condition (\ref{qdetCond}) is preserved by the transformation. Note that the authors of  \cite{R1} defined the quantum Geroch transformation only for  
the quantum monodromy matrix $\opM(u)$. Therefore, we must relate (\ref{qdetCond}) to some constraint written in terms of the quantum monodromy matrix.

Classically, the monodromy matrix $\clM(u)$ must have unit determinant. We suggest a quantum analogue of this constraint.
Let us consider the following expression:
\be
	\qM(u,s)=\opM(u)\check{\opM}(u+\ri\hbar+\ri s)\,.
\ee
In any representations satisfying requirements listed in the section \ref{section_rep} the operator $\qM(u,s)$ must have a pole at $s=0$ as long as both $u$ and $u+\ri\hbar$ are  within the intersection of the domains of analyticity of the quantum transition matrices $\opT_{+}(u)$ and $\opT_{-}(u)$.
The appearance of the pole of $\qM(u,s)$ at $s=0$, which does not appear in the classical case, can be informally understood as a manifestation of an infinite multiplicative renormalization factor.

We define the {\it regularized quantum determinant} of the quantum monodromy matrix as the residue of the top left component of $\mathfrak{M}$:
\be
	\qdet \opM(u) :=\frac{2}{3}{\rm Res}_{s=0} \qM^{11} (u,s)\,.
\ee
Then the following proposition holds.
\begin{proposition}
\label{prop_factor}
	Let both quantum transition matrices $\opT_{\pm}(u)$ be regular functions of the spectral parameter within the strip $-\textstyle\frac{1}{2}\hbar-\epsilon<\Im(u)<\textstyle\frac{1}{2}+\epsilon$ for some real positive $\epsilon$.
	Let all normal-ordered monomials composed of several instances of $\opT_{+}$ and $\opT_{-}$, i.e. expressions of the form
	\be
		\opT_{+}^{ij}(u_1)\opT_{+}^{ij}(u_2)\dots \opT_{+}^{ij}(u_k)
		\opT_{-}^{ij}(v_1)\opT_{-}^{ij}(v_2)\dots \opT_{+}^{ij}(v_m)\,
	\ee
	be regular as functions of all spectral parameters $u_i$, $v_i$ involved provided 
	all of these spectral parameters
	belong to the joint strip of analyticity. Then exchange relations imply the following identities:
\be
\label{MqdetRes}
	\qdet \opM(u)=\qdet \opT_{+}(u)\qdet \opT_{-}(u)\,,
\ee
\be
\label{MqdetMatrix}
	\frac{2}{3}{\rm Res}_{s=0} \qM^{ij} (u,s)=\qdet \opM(u)\delta_{ij}\,.
\ee
\end{proposition}
The conditions of the proposition are fulfilled if one adopts the set of assumptions outlined in the section \ref{section_rep}. 
The detailed derivation of the formulae (\ref{MqdetRes}),(\ref{MqdetMatrix}) can be found in the appendix.

The factorization formula (\ref{MqdetRes}) implies that the regularized quantum determinant of the monodromy matrix belongs to the center of the algebra $\mathcal{D}$.


\section{Quantum Geroch group}
\label{section_geroch}

The authors of \cite{R1} proposed the following quantization of the transformation (\ref{GerochMCl}): consider a quantum group defined by the relation
\be
	R(u-v)\overset{1}{\opS}(u)\overset{2}{\opS}(v)=\overset{2}{\opS}(v)\overset{1}{\opS}(u)R(u-v)\,,
\ee
and let it act on the quantum monodromy matrix $\opM(u)$ as
\be
\label{GerochMQ}
	\opM(u) \mapsto \tilde{\opM}(u)=\opS(u+\textstyle{\frac{\ri}{2}}\hbar)\opM(u)\opS(u-\textstyle{\frac{\ri}{2}}\hbar)\,,
\ee
assuming 
\be
	[\opM^{ij}(u),\opS^{kl}(v)]=0\,,\quad \forall u,v \in \mathbb{C}\,,\;i,j,k,l\in \{1,2\}\,.
\ee
Under the Hermitian conjugation the operator $\opS(u)$ transforms as follows:
\be
\label{sHConj}
	\left(\opS^{ij}(u)\right)^{\dag}=\opS^{ij}(\bar{u})\,.
\ee
The representation of Yangian in which the operator $\opS(u)$ is evaluated is not specified.
The work \cite{R1} contains the proof that the transformed quantum monodromy matrix $\tilde{\opM}(u)$ is symmetric and positive-definite for all real values of $u$.

The quantum Geroch group is not a symmetry group in the usual sense of the word.
More precisely, the transformation (\ref{GerochMQ}) is an isomorphism of the algebras generated by the operators $\opM(u)$ and $\tilde{\opM}(u)$ respectively, 
but if the operator $\opM(u)$ acts on the Hilbert space $\mathcal{H}$ and the operator $\opS(u)$ acts on some representation $\mathcal{G}$  of the Yangian  then the operator $\tilde{\opM}(u)$ acts on the tensor product $\mathcal{G}\otimes\mathcal{H}$.
The repeated application of the transformation (\ref{GerochMQ}) yields an operator acting on the tensor product $\mathcal{G}\otimes\mathcal{G}\otimes\mathcal{H}$, i.e. each copy of $\opS(u)$ acts on its own space.

The authors of \cite{R1} define the action of the quantum Geroch group on the  subalgebra of the algebra $\mathcal{D}$ generated by $M$.

We suggest the following  extension of the transformation (\ref{GerochMQ}) to the full algebra $\mathcal{D}$:
\bea
\label{GerochTQp}
	\opT_{+}(u)&\mapsto& \tilde{\opT}_{+}(u):= \opS(u\pm\textstyle{\frac{\ri}{2}}\hbar) \opT_{\pm}(u)\hat{\opU}(-u+\textstyle{\frac{\ri}{2}}\hbar)\,,\\[1mm]
\label{GerochTQm}
	\opT_{-}(u)&\mapsto& \tilde{\opT}_{-}(u):=\opS(u\pm\textstyle{\frac{\ri}{2}}\hbar) \opT_{\pm}(u)\opU^{T}(-u-\textstyle{\frac{\ri}{2}}\hbar)\,,
\eea
where the operator $\opU(u)$ generates another copy of Yangian. In other words, the operator $\opU(u)$ is an arbitrary operator satisfying the RTT relation:
\be
\label{UYang}
	R(u-v)\overset{1}{\opU}(u)\overset{2}{\opU}(v)=\overset{2}{\opU}(v)\overset{1}{\opU}(u) R(u-v)\,.
\ee
The symbols $\opS(u)$ and $\opU(u)$  are assumed to commute with the elements of the algebra $\mathcal{D}$ and between each other:
\bea
\label{GaugeCommST}
	[\opS^{ij}(u),\opT^{kl}_{\pm}(v)]&=&0\,,\\[1mm]
\label{GaugeCommSU}
	[\opS^{ij}(u),\opU^{kl}_{\pm}(v)]&=&0\,,\\[1mm]
\label{GaugeCommUT}
	[\opU^{ij}(u),\opT^{kl}_{\pm}(v)]&=&0\,.
\eea
The Hermitian conjugate of $\opU(u)$ is given by
\be
\label{uHConj}
	\left(\opU^{ij}(u)\right)^{\dag}=\epsilon_{ia}\epsilon_{jb}\opU^{ab}(\bar{u})\,.
\ee

The following properties can be established immediately.
\begin{proposition}
\label{propMtrf}
	The quantum monodromy matrix transforms according to the rule (\ref{GerochMQ}) under the transformation (\ref{GerochTQp}),(\ref{GerochTQm}).
\end{proposition}
\begin{proposition}
\label{propExRel}
	The transformed algebra $\tilde{\mathcal{D}}$ generated by the transformed transition matrices $\tilde{\opT}_{\pm}$ is isomorphic to the original algebra $\mathcal{D}$, i.e. the exchange relations are preserved:
\be
	R(u-v)\overset{1}{\tilde{\opT}}_{\pm}(u)\overset{2}{\tilde{\opT}}_{\pm}(v)=\overset{2}{\tilde{\opT}}_{\pm}(v)\overset{1}{\tilde{\opT}}_{\pm}(u)R(u-v)\,,
\ee
\be
	R(u-v-\ri\hbar)\overset{1}{\tilde{\opT}}_{-}(u)\overset{2}{\tilde{\opT}}_{+}(v)
	=\overset{2}{\tilde{\opT}}_{+}(v)\overset{1}{\tilde{\opT}}_{-}(u)R^{\eta}(u-v+\ri\hbar)\chi(u-v)\,.
\ee
\end{proposition}
\begin{proposition}
\label{propQdetFact}
The quantum determinant of the transformed quantum transition matrices factorizes as follows:
\be
	\qdet \tilde{\opT}_{\pm}(u)=\qdet \opS(u\pm \textstyle{\frac{\ri}{2}}\hbar) \qdet \opT_{\pm}(u) \qdet \opU (-u-\textstyle{\frac{\ri}{2}}\hbar)\,.
\ee
\end{proposition}
\begin{proposition}
\label{propHConj}
	The transformed transition matrices $\tilde{\opT}_{+}(u)$ and $\tilde{\opT}_{-}(u)$ are related via Hermitian conjugation:
	\be
		\left(\tilde{\opT}_{+}^{ij}(u)\right)^{\dag}=\tilde{\opT}_{-}^{ij}(\bar{u})\,.
	\ee
\end{proposition}

The commutation relations (\ref{GaugeCommST})-(\ref{GaugeCommUT}) imply that the operator $\opS(u)$ acts on some space $\mathcal{G}$, 
the operator $\opU(u)$ acts on another space $\mathcal{K}$, and the transformed transition matrices $\tilde{\opT}_{\pm}(u)$ act on the tensor product 
$\mathcal{G}\otimes\mathcal{H}\otimes\mathcal{K}$, where $\mathcal{H}$ is the Hilbert space of the untransformed model. The propositions \ref{propMtrf} - \ref{propQdetFact} hold for any choice of the representations $\mathcal{G}$ and $\mathcal{K}$ of the Yangian, and the proposition \ref{propHConj} is valid for any pair of representations satisfying identities (\ref{sHConj}), (\ref{uHConj}) respectively.

However, the resulting representation will not be ``physically acceptable'' for a generic choice of $\mathcal{G}$ and $\mathcal{K}$.
In order to get a ``physically acceptable'' representation, we must impose on the operator $\opU(u)$ some constraint
written in terms of the operators $\opS(u)$ and $\opT_{\pm}(u)$. The explicit formulation of this constraint is not known at the moment.

We can make the following observation, which helps to understand the difficulty with the formulation of this constraint.
\begin{proposition}
	\label{prop_nogo_hol}
	The transformed transition matrices $\tilde{\opT}_{\pm}(u)$ can't be analytical in the same half-planes as the original matrices.
\end{proposition}
\begin{proof}
	Suppose that for some positive number $\varepsilon$ the operator $\tilde{\opT}_{+}(u)$ is holomorphic in the half-plane $\Im(u)>-\textstyle{\frac{1}{2}}\hbar-\varepsilon$, and the operator $\tilde{\opT}_{-}(u)$ is holomorphic in the half-plane $\Im(u)<\textstyle\frac{1}{2}+\varepsilon$.
	Then any matrix element of any 
	of these operators must be a holomorpic function in the corresponding half-plane. Let us consider a pair of states $|\Omega\rangle$ and $|\Omega'\rangle$
	\bea
	|\Omega\rangle&=&|\phi\rangle\otimes|\psi\rangle\otimes |\chi\rangle\,,\\[1mm]
	|\Omega'\rangle&=&|\phi'\rangle\otimes|\psi'\rangle\otimes |\chi'\rangle\,,
	\eea
	with
	\bea
	|\phi\rangle,\; |\phi'\rangle \in \mathcal{G}\,,\\[1mm]
	|\psi\rangle,\; |\psi'\rangle \in\mathcal{H}\,,\\[1mm]
	|\chi\rangle,\; |\chi'\rangle \in \mathcal{K}\,.
	\eea
	The matrix elements $\langle \Omega|\tilde{\opT}^{ij}_{\pm}(u)|\Omega'\rangle$ then read
	\bea
		\langle \Omega|\tilde{\opT}^{ij}_{-}(u)|\Omega'\rangle&=&
		\langle \phi |\opS^{ia}(u-\textstyle{\frac{\ri}{2}}\hbar)|\phi'\rangle
		\langle \psi |\opT^{ab}_{-}(u)|\psi'\rangle
		\langle \chi |\opU^{jb}(u-\textstyle{\frac{\ri}{2}}\hbar)|\chi'\rangle \,,\\[1mm]
		\langle \Omega|\tilde{\opT}^{ij}_{+}(u)|\Omega'\rangle&=&
		\langle \phi |\opS^{ia}(u+\textstyle{\frac{\ri}{2}}\hbar)|\phi'\rangle
		\langle \psi |\opT^{ab}_{+}(u)|\psi'\rangle \epsilon_{bc}\epsilon_{jd}
		\langle \chi |\opU^{cd}(u-\textstyle{\frac{\ri}{2}}\hbar)|\chi'\rangle\,.
	\eea
	The problem of finding the matrix elements $\langle\chi|\opU^{ij}(u)|\chi'\rangle$ 
	as functions of the spectral parameter
	is similar to the problem we encounter in the classical case when we solve for the matrix $\clU\left(\clS,\clT_{\pm}\mid u\right)$.
	The only difference is that the determinant of the matrix 
	$\langle\chi|\opU(u)|\chi'\rangle$ is not fixed. 
	Still, the requirement that the matrix elements $\langle \Omega|\tilde{\opT}^{ij}_{\pm}(u)|\Omega'\rangle$ should be holomorphic in their respective half-plane significantly constraints possible choices for the matrix elements of the operator $\opU(u)$. In particular, they can have singularities only at the points where
	the matrix element $\langle\chi|\opS(u)|\chi'\rangle$ is singular or degenerate.
	
	Since the states $|\phi\rangle$, $|\phi'\rangle$, $|\psi\rangle$, $|\psi'\rangle$, $|\chi\rangle$ and $|\chi'\rangle$ can be chosen independently, any
	matrix element of $\opU(u)$ should be a common solution to the problems with  all pairs of states $|\phi\rangle$, $|\phi'\rangle$ from $\mathcal{G}$.
	It is possible only if all matrix elements of the operator $\opS(u)$ have their singularities at the same points.
	But that is incompatible with the classical limit, as almost any meromorphic function with unit determinant and an appropriate asymptotic behaviour 
	represents an element of the classical Geroch group.
\end{proof}
The proposition \ref{prop_nogo_hol} means that the operator $\opU(u)$ must be fixed by some weaker condition.
This hypothetical condition should still give rise to a ``physically acceptable'' representation, albeit in an indirect way.
We suggest the following scenario, which is not forbidden by the proposition \ref{prop_nogo_hol}: the representation of $\mathcal{D}$ on the space $\,\mathcal{G}\otimes\mathcal{H}\otimes\mathcal{K}\,$ decomposes into a direct sum of multiple representations.
The operator $\opU(u)$ should be chosen in such a way that at least one of these subrepresentations is ``physically acceptable''.
In this scenario the quantum Geroch transformation (\ref{GerochTQp}),(\ref{GerochTQm}) is not an automorphism of the ``physically acceptable'' representation, but it can induce such automorphisms. The latter are constructed by identifying the vacuum state from  the space $\mathcal{H}$ with some vector from one of the ``physically acceptable'' subrepresentations.

The extension of the quantum Geroch group described above may be used to prove that the transformed monodromy matrix satisfies the identity (\ref{MqdetRes}). Indeed, by construction the transformed monodromy matrix 
admits factorization
\be
	\tilde{\opM}=\tilde{\opT}_{+}(u)\tilde{\opT}_{-}^{T}(u)
\ee
such that the transformed transition matrices $\tilde{\opT}_{\pm}$ give a representation of the deformed quantum double (\ref{RTTint}),(\ref{Rpmint}) and 
\be
	\text{qdet}\tilde{\opT}_{\pm}(u)=1\,.
\ee
The derivation of the factorization formula (\ref{MqdetRes}) remains valid for the transformed transition matrices provided that the assumptions about analyticity of the transition matrices and their normal-ordered products hold for the transformed algebra.
These assumptions would be fulfilled for the restriction of the transformed transition matrices onto one of the ``physically acceptable'' subrepresentations in the hypothetical scenario outlined above.
If that is the case, the regularized quantum determinant of the quantum monodromy matrix is an invariant of the  quantum Geroch group. 
\section{Conclusion}

In this paper we constructed a regularization of the quantum determinant of the quantum monodromy matrix  which deforms the determinant of the classical monodromy matrix.
We expressed it as a product of quantum determinants of the transition matrices.
We gave arguments supporting the conjecture that the regularized quantum determinant of the quantum  monodromy matrix is 
preserved by the action of the quantum Geroch group. 
We also suggested the extension of the action of the quantum dressing transformation onto the full algebra of observables which preserves the quantum determinants of the quantum transition matrices.
To prove that the action of the quantum Geroch group is an automorphism of the representation it remains to verify the analyticity assumptions for both original 
and transformed transition matrices.

The quantum integrable structure of \ER\ merits further investigation.
To begin with, the complete description of the Hilbert space of the system is still unknown, as well as the explicit action of the quantum Geroch group on the space of states.
Furthermore, there is no known  infinite family of commuting integrals of motion for the Einstein-Rosen model, even though the presence of such families is a common feature of integrable quantum models.
Finally, an explicit expression of Hamiltonian  in terms of quantum transition matrices is also missing.

\section*{Acknowledgements}
I am grateful to Dmitry Korotkin for introducing me to this problem, and to Michael Reisenberger for his interest in my work and insightful discussions.
\appendix
\section{Proof of the factorization formula for residue of quantum determinant of $\opM(u)$}
\label{appA}
In components, the matrix $\qM(u,s)$ can be written as
\be
	\qM^{ij}(u,s)=\epsilon_{j\nu}\epsilon_{\alpha\beta}\opM^{i\alpha}(u)\opM^{\nu\beta}(\hat{u}+\ri s)\,,
\ee
where we use the abbreviated notation
\be
	\hat{u}=u+\ri\hbar\,.
\ee
Using the exchange relation (\ref{Rpmint}) we rewrite the operator $\qM(u,s)$ in its normal-ordered form:
\bea
	\qM^{ij}(u,s)&=&\epsilon_{j\nu}\epsilon_{\alpha\beta}\opT_{+}^{ia}(u)\opT_{-}^{\alpha a}(u)\opT_{+}^{\nu b}(\hat{u}+\ri s)\opT_{-}^{\beta b}(\hat{u}+\ri s)\nonumber\\
	&=&\epsilon_{j\nu}\epsilon_{\alpha\beta}\opT_{+}^{ia}(u)R^{-1}_{\alpha k \nu l}(-2\ri\hbar-\ri s)\opT_{+}^{l l'}(\hat{u}+\ri s)\\
	&\times& \opT_{-}^{kk'}(u) R^{\eta}_{k' a l' b}(-\ri s) \opT_{-}^{\beta b}(\hat{u}+\ri s) \chi(-\ri\hbar-\ri s)\,.\nonumber
\eea
Substituting the explicit expressions for $\chi$, $R$ and $R^{\eta}$ we obtain
\bea
	\qM^{ij}(u,s)&=&\frac{\epsilon_{j\nu}\epsilon_{\alpha\beta}}{s(s+2\hbar)}
	\opT_{+}^{ia}(u)\left((2\hbar+s)\delta_{\alpha k}\delta_{\nu l}-\hbar\delta_{\alpha l}\delta_{\nu k}\right)\nonumber\\
	&\times& \opT_{+}^{ll'}(\hat{u}+\ri s) \opT_{-}^{kk'}(u)\left((\hbar+s)\delta_{k' a}\delta_{l' b}-\hbar\delta_{k'l'}\delta_{ab}\right)
	\opT_{-}^{\beta b}(\hat{u}+\ri s)\,.
\eea
Assuming the representation satisfies the conditions outlined in the section \ref{section_rep} we compute the limit of operator $s\qM(u,s)$ as $s$ tends to zero:
\bea
\label{qMlimitTens}
	\lim\limits_{s\rightarrow 0} s\qM^{ij}(u,s)&=&\frac{\hbar\epsilon_{j\nu}\epsilon_{\alpha\beta}}{2}
	\opT_{+}^{ia}(u)(2\delta_{\alpha k}\delta_{\nu l}-\delta_{\alpha l}\delta_{\nu k})
	\opT_{+}^{lb}(\hat{u})\nonumber\\
	 &\times & \opT_{-}^{kk'}(u)(\delta_{k' a}\delta_{ b l'}-\delta_{k'b}\delta_{al'})
	\opT_{-}^{\beta l'}(\hat{u})\,.
\eea
In the above expression we relabeled the contracted indices to make use of the following identity for $2\times 2 $ matrices:
\bea
\label{epsid}
	 \epsilon_{ac}\epsilon_{bd}\opT_{+}^{ia}(u)\opT_{+}^{\alpha c}(v)\opT_{-}^{\nu b}(u)\opT_{-}^{\beta d}(v)
	 &=&\opT_{+}^{ia}(u)\opT_{+}^{\alpha b}(v)\opT_{-}^{\nu a}(u)\opT_{-}^{\beta b}(v)\nonumber\\
	 &\;&-\opT_{+}^{ia}(u)\opT_{+}^{\alpha b}(v)\opT_{-}^{\nu b}(u) \opT_{-}^{\beta a}(v)\,.
\eea
The identity (\ref{epsid}) is verified by explicitly calculating  both sides. Combined with the definition of the quantum determinants of the quantum transition matrices $T_{\pm}$ in the following form
\be
	\qdet \opT_{\pm}(u)\epsilon_{ij}=\epsilon_{ab}\opT^{ia}_{\pm}(u)\opT^{jb}_{\pm}(\hat{u})
\ee
it allows to simplify the expression (\ref{qMlimitTens}) further to get
\be
	\lim\limits_{s\rightarrow 0} s \qM^{ij}(u,s)=\frac{3}{2}\qdet \opT_{+}(u)\qdet \opT_{-}(u) \delta_{ij}
\ee

\end{document}